%% file: 05-10-2026.tex
\documentclass[12pt,a4paper]{article}
\input{preamble}
\newcommand{\SP}{\mathcal{S}}
\begin{document}

\title{A Characterization of Quantile Stable Matchings in Many-to-Many Markets\thanks{%Thanks \ldots. 
We acknowledge financial support
from UNSL through grants 032016 and 030320, from Consejo Nacional
de Investigaciones Cient\'{\i}ficas y T\'{e}cnicas (CONICET) through grant
PIP 112-200801-00655, and from Agencia Nacional de Promoción Cient\'ifica y Tecnológica through grant PICT 2017-2355.}}


\author{Paola Manasero\samethanks [2]\thanks{Instituto de Matem\'{a}tica Aplicada San Luis (UNSL-CONICET) and Departamento de Matemática, Universidad Nacional de San
Luis, San Luis, Argentina. Emails:  \texttt{pbmanasero@unsl.edu.ar} (P. Manasero) and \texttt{pabloneme08@gmail.com} (P. Neme).}    \and Pablo Neme\samethanks[2]}

\date{\today}
\maketitle

\begin{abstract}
We study two-sided many-to-many matching markets where agents' preferences satisfy strong substitutability and the law of aggregate demand (LAD). In this framework, the set of stable matchings forms a finite distributive lattice, and quantile stable matchings exist and satisfy the polarization of interests. Building on the cycle-reduction machinery of \cite{bonifacio2022cycles}, we provide a purely combinatorial characterization of quantile stable matchings formulated in terms of preference cycles and their induced precedence order. Specifically, we define an order-theoretic counting function on cycles and prove that each quantile stable matching corresponds to the elimination of an endogenously determined closed set of cycles starting from the firm-optimal stable matching. This characterization extends the rotation-based framework of \cite{cheng2008generalized} from one-to-one markets to many-to-many markets and bypasses the need for exhaustive outcome enumeration or partner sorting across the entire stable set.
\bigskip

\noindent \emph{JEL classification:} C78, D47.\bigskip

\noindent \emph{Keywords:} Many-to-many matching, Quantile stable matchings, Preference cycles, Distributive lattices, Birkhoff's representation theorem.

\end{abstract}
\section{Introduction}
In two-sided matching markets, the Deferred Acceptance algorithm of \cite{gale1962college} reliably outputs stable allocations, but it does so by reaching the extreme boundaries of the stable set: the firm-optimal stable matching (which is unanimously the worst stable outcome for workers) and the worker-optimal stable matching (unanimously the worst for firms). While computationally appealing, this stark polarization of interests raises significant fairness concerns in centralized clearinghouses, where prioritizing one market side imposes maximal welfare losses on the other. 

To overcome this inherent one-sided bias and provide balanced compromises, \cite{teo1998geometry} pioneered the study of \emph{generalized median stable matchings} in one-to-one markets. They established that selecting each agent's quantile-ranked stable partner across the entire set of stable matchings yields an allocation that is itself stable. This formulation introduces a flexible spectrum of egalitarian and compromise solutions between the two market extremes. In particular, the standard median matching treats both sides symmetrically and mitigates the vulnerability of agents to strategic manipulation often observed in extremal mechanisms \citep{klaus2006median, chen2016median, chen2021quantile}, although recent work highlights subtle incentive trade-offs regarding obvious manipulations in contractual settings \citep{arribillaga2025note}.

Moving beyond the one-to-one framework, \cite{chen2016median} extended this compromise paradigm to many-to-many markets under the banner of \emph{quantile stable matchings}. They showed that whenever agents' preferences satisfy the law of aggregate demand (LAD) and \emph{strong substitutability}---a refinement of standard substitutability introduced by \cite{echenique2006theory}---simultaneously assigning to each firm its ranked stable set of workers produces an allocation that is itself stable. Moreover, these compromise matchings satisfy a strict duality of interests across the distributive lattice: a favorable choice for firms corresponds precisely to the dual, symmetric outcome for workers.

Traditionally, defining or computing these quantile stable matchings requires an exhaustive global inspection of the entire stable matching set to identify each agent's ranked stable outcomes. In the one-to-one domain, \cite{cheng2008generalized} overcame this limitation by showing that generalized median matchings can instead be characterized through an ideal-counting argument on a precedence order of rotations---the algorithmic terminology introduced by \cite{irving1986complexity} for cycles in preferences \citep{roth1992two}. In multi-partner contexts, however, a purely structural characterization that identifies these compromise allocations without exhaustive enumeration of the stable set has remained absent. This fundamental contrast naturally raises a central question: \emph{Can quantile stable matchings in many-to-many markets be identified through an intrinsic combinatorial structure without enumerating the entire set of stable matchings?}

In this note, we answer this question affirmatively by extending the combinatorial perspective to many-to-many markets using the concept of \emph{cycles in preferences} introduced by \cite{bonifacio2022cycles}. In their framework, each stable matching adjacent in Blair's order is obtained from another by eliminating a cycle in preferences under a reduced preference profile. We show that the set of all such cycles, endowed with their natural precedence ordering, allows us to characterize every quantile stable matching directly: each compromise matching corresponds to an endogenously defined closed set of cycles, identified solely by how many stable matchings leave each cycle uneliminated.

 The remainder of the note is organized as follows. Section \ref{sec:model} introduces the model, the preference conditions, and the distributive lattice structure of stable matchings. Section \ref{sec:cycles} reviews the preference reduction procedure, cycles in preferences, and their lattice-theoretic connection via Birkhoff's representation theorem. Section \ref{sec:characterization} presents our main combinatorial characterization of quantile stable matchings and its proof. Section \ref{sec:conclusion} concludes with a discussion on computational complexity and directions for future research.
\section{Model and Preliminaries}\label{sec:model}

Let $F$ and $W$ be disjoint, finite sets of firms and workers. Let $A = F \cup W$. Each firm $f \in F$ is endowed with a strict preference relation $P_f$ over $2^W$, and each worker $w \in W$ with a strict preference relation $P_w$ over $2^F$. A preference profile is denoted by $P = (P_a)_{a \in A}$. For any $f \in F$ and $S \subseteq W$, $C_f(S)$ denotes $f$'s chosen subset from $S$ according to $P_f$; similarly, $C_w(S')$ denotes $w$'s chosen subset from $S' \subseteq F$.

A matching $\mu$ is a mapping from $F \cup W$ into $2^{F \cup W}$ such that for all $f \in F$ and $w \in W$: (i) $\mu(f) \subseteq W$, (ii) $\mu(w) \subseteq F$, and (iii) $w \in \mu(f) \iff f \in \mu(w)$.

A matching $\mu$ is \emph{individually rational} if $C_a(\mu(a)) = \mu(a)$ for all $a \in A$. A pair $(f, w) \in F \times W$ blocks $\mu$ if $w \notin \mu(f)$, $w \in C_f(\mu(f) \cup \{w\})$, and $f \in C_w(\mu(w) \cup \{f\})$. A matching $\mu$ is \emph{stable} if it is individually rational and admits no blocking pair. We denote the set of stable matchings by $\SP$.

\begin{definition}
A preference relation $P_a$ satisfies:
\begin{enumerate}
 \item \textbf{Strong Substitutability} if for any $S, S'$ such that $C_a(S') P_a C_a(S)$ and any $b$, 
\[
b \in C_a(S' \cup \{b\}) \implies b \in C_a(S \cup \{b\}).
\footnote{Strong substitutability \citep{echenique2006theory} strictly refines the standard substitutability condition \citep{kelso1982job}, which requires that for all $S \subseteq S'$, $b \in C_a(S' \cup \{b\}) \implies b \in C_a(S \cup \{b\})$. Since $S \subseteq S'$ implies $C_a(S') P_a C_a(S)$, any strongly substitutable preference relation is automatically substitutable.}\]
    \item \textbf{The Law of Aggregate Demand (LAD)} if for all $S \subseteq S'$, $|C_a(S)| \le |C_a(S')|$.
    
\end{enumerate}
\end{definition}

Given an agent's strict preference relation, \cite{blair1988lattice} defined a partial order over subsets of the opposite side of the market induced by the agent's choice. For any firm $f \in F$ with preference relation $P_f$ and any two subsets of workers $S, S' \subseteq W$, we say that $S$ is \emph{Blair-preferred} to $S'$, denoted by $S \ge_f S'$, if 
\[
S = C_f(S \cup S').
\]
That is, $S \ge_f S'$ means that whenever all workers in $S \cup S'$ are available, the most-preferred subset of firm $f$ according to $P_f$ is exactly $S$. We write $S >_f S'$ if $S \ge_f S'$ and $S \ne S'$. Symmetrically, for each worker $w \in W$ and subsets $F', F'' \subseteq F$, the partial order $\ge_w$ is defined by $F' \ge_w F''$ if and only if $F' = C_w(F' \cup F'')$.

At the market level, Blair's order induces a unanimous partial ordering over matchings: for any two matchings $\mu$ and $\mu'$, we write $\mu \ge_F \mu'$ if $\mu(f) \ge_f \mu'(f)$ for all $f \in F$, and $\mu >_F \mu'$ if $\mu \ge_F \mu'$ with $\mu \ne \mu'$. The unanimous order for workers, $\ge_W$, is defined symmetrically. Under substitutable preferences (and in particular under strongly substitutable preferences), \cite{blair1988lattice} established that the set of stable matchings $\mathcal{S}$ forms a lattice with respect to $\ge_F$ (and dually with respect to $\ge_W$). 

When preferences further satisfy LAD, $(\mathcal{S}, \ge_F)$ forms a finite distributive lattice and these lattices are dual, giving rise to the polarization of interests: $\mu \ge_F \mu'$ if and only if $\mu' \ge_W \mu$ for all $\mu, \mu' \in \mathcal{S}$ \citep{alkan2002class}. Furthermore, under strong substitutability, Blair's order on $\mathcal{S}$ coincides with the underlying strict individual preference relation $P_a$ over each agent's stable assignments.

\section{Cycles in Preferences and Closed Sets of Cycles}\label{sec:cycles}
In this section, we review the notion of cycles in preferences introduced by \cite{bonifacio2022cycles}. In their framework, transitions between adjacent stable matchings in the lattice are governed by eliminating elementary cycles of worker-firm pairs under reduced preference profiles. For completeness, the formal construction of the preference reduction procedure is deferred to Appendix \ref{app:reduction}.

For $\mu, \tilde{\mu} \in \SP$ with $\mu >_F \tilde{\mu}$, the reduced preference profile $P^{\mu, \tilde{\mu}}$ eliminates partners that are strictly Blair-preferred to $\mu$ or strictly less Blair-preferred than $\tilde{\mu}$. When $\tilde{\mu} = \mu_W$, we denote this profile by $P^\mu$.

\begin{definition}[Cycle in Preferences]
Let $\mu \in \SP \setminus \{\mu_W\}$. A cycle $\sigma$ for $P^\mu$ is an ordered sequence of worker-firm pairs $\sigma = \{(w_1, f_1), (w_2, f_2), \dots, (w_r, f_r)\}$ such that for all $i = 1, \dots, r$:
\begin{enumerate}
    \item $w_i \in \mu(f_i) \setminus \mu_W(f_i)$,
    \item $C_{f_i}^\mu(W \setminus \{w_i\}) = (\mu(f_i) \setminus \{w_i\}) \cup \{w_{i+1}\}$, with $w_{r+1} = w_1$, and
    \item $C_{w_i}^\mu(\mu(w_i) \cup \{f_{i-1}\}) = (\mu(w_i) \setminus \{f_i\}) \cup \{f_{i-1}\}$, with $f_0 = f_r$.
\end{enumerate}
\end{definition}

Eliminating a cycle $\sigma$ from matching $\mu$ yields the cyclic matching $\mu_\sigma$, defined by $\mu_\sigma(f) = (\mu(f) \setminus \{w_i\}) \cup \{w_{i+1}\}$ for each $f = f_i \in \sigma$, and $\mu_\sigma(f) = \mu(f)$ otherwise. \cite{bonifacio2022cycles} prove that $\mu_\sigma \in \SP$ and $\mu >_F \mu_\sigma$. In terms of the lattice $(\SP, \ge_F)$, this cyclic transition defines an immediate covering relation: there is no stable matching $\mu'' \in \SP$ such that $\mu >_F \mu'' >_F \mu_\sigma$. Thus, we say that $\mu$ \emph{covers} $\mu_\sigma$.

Let $\Sigma$ denote the set of all elementary preference cycles exposed across the lattice $\SP$.

\begin{definition}[Precedence of Cycles]
For two cycles $\sigma, \sigma' \in \Sigma$, we say that $\sigma$ \emph{precedes} $\sigma'$ (denoted $\sigma \le \sigma'$) if cycle $\sigma$ must be eliminated before cycle $\sigma'$ can become exposed in a reduced preference profile along any path from $\mu_F$. 

A subset $S \subseteq \Sigma$ is a \emph{closed set of cycles} (also known in order theory as an \emph{order ideal} or \emph{down-set}) if whenever $\sigma \in S$ and $\sigma' \le \sigma$, then $\sigma' \in S$. We denote the family of all closed sets of cycles by $\mathcal{C}(\Sigma)$.
\end{definition}

By Birkhoff's Representation Theorem for finite distributive lattices,\footnote{Birkhoff's Representation Theorem \citep{birkhoff1940lattice} establishes that any finite distributive lattice $(L, \le)$ is isomorphic to the lattice of down-sets (or order ideals) $\mathcal{J}$ of its partially ordered set of join-irreducible elements $P$.} in our setting the covering relations in $(\mathcal{S}, \ge_F)$ are characterized by the elimination of elementary preference cycles from $\Sigma$. 

Because we initialize this elimination procedure at the firm-optimal stable matching $\mu_F$ (the top element of the lattice), each eliminated cycle moves the matching downward along the covering relations of Blair's order $\ge_F$. Consequently, the lattice of stable matchings $(\mathcal{S}, \ge_F)$ is isomorphic to the family of closed sets of cycles $(\mathcal{C}(\Sigma), \subseteq)$ inverted. More formally, there exists a bijective correspondence between $\mathcal{S}$ and $\mathcal{C}(\Sigma)$ such that:
\begin{enumerate}
    \item The firm-optimal matching $\mu_F$ corresponds to the empty set of cycles, $\emptyset$.
    \item The worker-optimal matching $\mu_W$ corresponds to the full set of cycles, $\Sigma$.
    \item For any two stable matchings $\mu, \mu' \in \mathcal{S}$ with corresponding closed sets $S_\mu, S_{\mu'} \in \mathcal{C}(\Sigma)$,
    \[
    \mu \ge_F \mu' \iff S_\mu \subseteq S_{\mu'}.
    \]
    In particular, every stable matching $\mu \in \mathcal{S}$ is uniquely obtained by starting at $\mu_F$ and eliminating all cycles belonging to its corresponding closed set $S_\mu$.
\end{enumerate}

\section{Characterization of Quantile Stable Matchings}
\label{sec:characterization}

Let $N = |\mathcal{S}|$ denote the total number of stable matchings. For each agent $a \in F \cup W$, consider the collection of subsets of partners assigned to $a$ across all stable matchings, $\{\mu(a) : \mu \in \mathcal{S}\}$. We sort these sets from most preferred to least preferred according to the strict preference relation $P_a$:
\begin{equation*}
\mu^{(1)}(a) >_a \mu^{(2)}(a) >_a \dots >_a \mu^{(N)}(a),
\end{equation*}
where the parenthesized superscript $(i)$ indicates an order statistic: $\mu^{(i)}(a)$ represents the $i$-th most-preferred set of partners received by agent $a$ across the stable set $\mathcal{S}$. 

We now formalize the market-wide compromise matchings introduced by \cite{chen2016median}, where a superscript $i$ (without parentheses) indexes the resulting matching.

\begin{definition}[Quantile Stable Matching]
\label{defino quantile stable}
For each integer $i \in \{1, \dots, N\}$, the $i$-th quantile stable matching for firms, denoted by $\mu_F^i$, is defined by:
\begin{equation*}
\mu_F^i(f) = \mu^{(i)}(f) \quad \text{for all } f \in F.
\end{equation*}
Symmetrically, the $i$-th quantile stable matching for workers, denoted by $\mu_W^i$, is defined by:
\begin{equation*}
\mu_W^i(w) = \mu^{(i)}(w) \quad \text{for all } w \in W.
\end{equation*}
\end{definition}
In particular, the extreme quantiles coincide with the lattice boundaries: $\mu_F^1$ is the firm-optimal stable matching $\mu_F$ (where each firm $f$ receives its first-choice stable set of workers $\mu^{(1)}(f)$), whereas $\mu_F^N$ is the worker-optimal stable matching $\mu_W$ (the least-preferred stable outcome for every firm).
Under strong substitutability and LAD, Theorem 2 in \cite{chen2016median} guarantees that $\mu_F^i$ is indeed a stable matching for each $i \in \{1, \dots, N\}$, and satisfies the polarization of interests: $\mu_F^i = \mu_W^{N-i+1}$.

We now characterize these matchings directly through the structure of preference cycles.

\begin{definition}
For each preference cycle $\sigma \in \Sigma$, let $n_\sigma$ denote the number of closed sets of cycles that do not contain $\sigma$:
\[
n_\sigma = |\{ S \in \mathcal{C}(\Sigma) : \sigma \notin S \}|.
\]
\end{definition}

\begin{lemma}\label{lemma sobre el numero de closed sets}
The function $n_\sigma$ is strictly increasing in the cycle precedence order: for any $\sigma,\sigma'\in \Sigma$, $\sigma < \sigma'$ implies that $n_\sigma < n_{\sigma'}$.
\end{lemma}
\begin{proof}
Let $S \in \mathcal{C}(\Sigma)$ and let $\sigma, \sigma' \in \Sigma$ with $\sigma < \sigma'$. Since $S$ is closed under precedence, $\sigma' \notin S$ implies $\sigma \notin S$. Thus,
\[
\{S \in \mathcal{C}(\Sigma) : \sigma' \notin S\} \subseteq \{S \in \mathcal{C}(\Sigma) : \sigma \notin S\}.
\]
To show that the inclusion is strict, consider the set $S_\sigma = \{\sigma'' \in \Sigma : \sigma'' \le \sigma\}$. By definition of precedence, $S_\sigma$ is a closed set of cycles containing $\sigma$. Moreover, since $\sigma < \sigma'$ and the precedence order is strict, $\sigma' \notin S_\sigma$. Thus, $S_\sigma$ belongs to the right-hand set but not to the left-hand set, which implies $n_\sigma < n_{\sigma'}$.
\end{proof}

\begin{theorem}\label{teorema de eliminar ciclos es un matching quantil estable}
Assume that agents' preferences satisfy strong substitutability and the law of aggregate demand (LAD). For each $i \in \{1, \dots, N\}$, define:
\[
S(i) = \{ \sigma \in \Sigma : n_\sigma < i \}.
\]
Then:
\begin{enumerate}[(i)]
    \item $S(i)$ is a closed set of cycles in $(\Sigma, \le)$.
    \item The $i$-th quantile stable matching $\mu_F^i$ is precisely the stable matching obtained by eliminating all cycles in $S(i)$ starting from $\mu_F$.
\end{enumerate}
\end{theorem}
\begin{proof}
To prove (i), let $\sigma \in S(i)$ and suppose $\sigma' \le \sigma$ for $\sigma'\in \Sigma$. By Lemma \ref{lemma sobre el numero de closed sets}, $n_{\sigma'} \le n_\sigma < i$, which implies that $\sigma' \in S(i)$. Therefore, $S(i)$ is a closed set of cycles in $(\Sigma, \le)$.

To prove (ii), let $S(\mu_F^i)$ denote the unique closed set of cycles in $\mathcal{C}(\Sigma)$ corresponding to the $i$-th quantile stable matching $\mu_F^i$. We must prove that $S(\mu_F^i) = S(i)$. 

Recall that for each cycle $\sigma \in \Sigma$, $n_\sigma$ counts the number of stable matchings whose corresponding closed set does not contain $\sigma$:
\[
n_\sigma = |\{ \mu \in \mathcal{S} : \sigma \notin S_\mu \}|.
\]
By Definition \ref{defino quantile stable} (whose stability is guaranteed by Theorem 2 in \cite{chen2016median}), each firm $f \in F$ receives under $\mu_F^i$ its $i$-th preferred stable set of workers $\mu(f)^{(i)}$. Since preferences are strict and aligned with the lattice operations under strong substitutability, a cycle $\sigma$ has been eliminated at $\mu_F^i$ (that is, $\sigma \in S(\mu_F^i)$) if and only if the number of stable matchings that dominate the execution of $\sigma$ is strictly less than $i$.

Indeed, if $n_\sigma < i$, there are fewer than $i$ stable matchings in which $\sigma$ remains uneliminated; thus, across the sorted list of $N$ stable assignments, every agent involved in $\sigma$ must have already traversed the transition induced by $\sigma$ before or at their $i$-th choice, implying $\sigma \in S(\mu_F^i)$. Conversely, if $n_\sigma \ge i$, there are at least $i$ stable matchings where $\sigma$ has not yet been executed, which means that at the $i$-th preferred stable outcome for firms, the cycle $\sigma$ has not been eliminated, so $\sigma \notin S(\mu_F^i)$.

Consequently, a cycle $\sigma$ is eliminated in $\mu_F^i$ precisely when $n_\sigma < i$. This establishes that $S(\mu_F^i)$ and $S(i)$ contain the exact same cycles and, therefore,  $S(\mu_F^i) = S(i)$.
\end{proof}

\section{Conclusion}\label{sec:conclusion}
In this paper, we have established a combinatorial characterization of quantile stable matchings in two-sided many-to-many matching markets where agents' preferences satisfy strong substitutability and the law of aggregate demand. By relying exclusively on the notion of cycles in preferences and their induced precedence order, we demonstrate that each quantile stable matching is uniquely determined by the elimination of a closed set of cycles, identified via an ideal-counting condition.

This characterization extends the classic rotation-based framework of \cite{cheng2008generalized} from one-to-one markets to many-to-many environments without imposing artificial chain structures on individual agents. Importantly, it reveals that the compromise matchings across the entire distributive lattice are fully encoded within the poset of preference cycles, bypassing the need for exhaustive outcome enumeration or agent-by-agent partner sorting. 

A remark on computational complexity is in order. In many-to-many markets, dealing with preferences over sets of partners requires access to a choice oracle, as in \cite{bonifacio2022cycles}. Even with such an oracle, computing exact quantile stable matchings is computationally demanding, as the underlying problem of counting ideals in a poset is known to be \#P-complete \citep{irving1986complexity, cheng2008generalized}. Whether this cycle-based framework can facilitate efficient approximations or heuristics for compromise stable matchings remains an interesting open question.

Beyond computational aspects, our combinatorial identification of quantile stable matchings also provides an operational tool for mechanism designers evaluating the balance between fairness and incentive compatibility. While quantile stable matchings offer appealing egalitarian compromises, recent findings by \cite{arribillaga2025note} demonstrate that any non-extremal quantile mechanism becomes susceptible to obvious manipulations when contractual terms are present. Understanding how our structural cycle representation interacts with strategic bidding, manipulation incentives, or constrained choice domains constitutes a natural avenue for future research.
\appendix
\section{Appendix: The Preference Reduction Procedure}
\label{app:reduction}

Given a strongly substitutable preference profile and two Blair-comparable (for the firms) stable matchings, this reduction procedure generates a new preference profile in which the most Blair-preferred stable matching is the firm-optimal matching and the least Blair-preferred stable matching is the worker-optimal matching, for the market identified with this new preference profile \citep{bonifacio2022cycles}. The reduction procedure is described as follows.

Let $\mu$ and $\tilde{\mu}$ be stable matchings such that $\mu \ge_F \tilde{\mu}$.

\begin{enumerate}
\item[\textbf{Step 1:}] 
\begin{enumerate}[(a)]
\item For each $f \in F$, each $W' \subseteq W$ such that $W' >_f \mu(f)$, and each $\tilde{w} \in W' \setminus \mu(f)$, remove each $\widetilde{W} \subseteq W$ such that $\tilde{w} \in \widetilde{W}$ from $f$'s list of acceptable sets of workers.

\item For each $w \in W$, each $F' \subseteq F$ such that $F' >_w \tilde{\mu}(w)$, and each $\tilde{f} \in F' \setminus \tilde{\mu}(w)$, remove each $\widetilde{F} \subseteq F$ such that $\tilde{f} \in \widetilde{F}$ from $w$'s list of acceptable sets of firms.
\end{enumerate}

\item[\textbf{Step 2:}] 
\begin{enumerate}[(a)]
\item For each $f \in F$, each $W' \subseteq W$ such that $\tilde{\mu}(f) >_f W'$, and each $\tilde{w} \in W' \setminus \tilde{\mu}(f)$, remove each $\widetilde{W} \subseteq W$ such that $\tilde{w} \in \widetilde{W}$ from $f$'s list of acceptable sets of workers.

\item For each $w \in W$, each $F' \subseteq F$ such that $\mu(w) >_w F'$, and each $\tilde{f} \in F' \setminus \mu(w)$, remove each $\widetilde{F} \subseteq F$ such that $\tilde{f} \in \widetilde{F}$ from $w$'s list of acceptable sets of firms.
\end{enumerate}

\item[\textbf{Step 3:}] After Steps 1 and 2 are performed, if $f$ is not acceptable for $w$ (that is, if $\{f\}$ is not on $w$'s preference list as now modified), remove each $W' \subseteq W$ such that $w \in W'$ from $f$'s list of acceptable sets of workers. If $w$ is not acceptable for $f$ (that is, if $\{w\}$ is not on $f$'s preference list as now modified), remove each $F' \subseteq F$ such that $f \in F'$ from $w$'s list of acceptable sets of firms.
\end{enumerate}

The profile obtained by this procedure is called the \textbf{reduced preference profile with respect to $\boldsymbol{\mu}$ and $\boldsymbol{\tilde{\mu}}$}, and is denoted by $\boldsymbol{P^{\mu,\tilde{\mu}}}$. When $\tilde{\mu} = \mu_W$, the profile is simply called the \textbf{reduced preference profile with respect to $\boldsymbol{\mu}$}, and is denoted by $\boldsymbol{P^{\mu}}$.

\end{document}

%% file: preamble.tex
\usepackage[utf8]{inputenc}
\usepackage[T1]{fontenc}
\usepackage{mathpazo}

\usepackage{amsmath,amssymb,amsfonts,amsthm}
\usepackage{mathrsfs}
\usepackage{nccmath}
\usepackage{cancel}

\usepackage{booktabs}
\usepackage{enumerate}
\usepackage{geometry}
\usepackage{graphicx}
\usepackage[usenames]{color}
\usepackage{soul}
\usepackage{emptypage}
\usepackage{natbib}
\usepackage{fancyhdr}

\usepackage{pgf,tikz}
\usetikzlibrary{intersections}
\usetikzlibrary{decorations.pathreplacing,angles,quotes}
\usetikzlibrary{arrows.meta}
\usetikzlibrary{arrows,decorations.markings}

\tikzset{
  myptr/.style={
    decoration={
      markings,
      mark=at position 1 with {\arrow[scale=2,>=stealth]{>}}
    },
    postaction={decorate}
  }
}

\usepackage{hyperref}
\hypersetup{
    colorlinks=true,
    citecolor=blue,
    filecolor=black,
    linkcolor=blue,
    urlcolor=blue
}

\newcommand*\samethanks[1][\value{footnote}]{\footnotemark[#1]}

\DeclareFontFamily{T1}{calligra}{}
\DeclareFontShape{T1}{calligra}{m}{n}{<-> s*[1.44] callig15}{}
\DeclareMathAlphabet{\mathcalligra}{T1}{calligra}{m}{n}

\newtheorem{theorem}{Theorem}

\newtheorem{definition}{Definition}

\newtheorem{example}{Example}

\newtheorem{lemma}{Lemma}

\renewenvironment{proof}[1][Proof]{\noindent \emph{#1.} }{\hfill$\square$\vspace{5pt}}